\documentclass[12pt]{article}

\usepackage{badescu}
\usepackage[toc]{appendix}

\newcommand*{\thanksNote}{%
  \thanks{Computer Science Department, Carnegie Mellon University.
    Supported by NSF grant FET-1909310. This material is based upon work
    supported by the National Science Foundation under grant numbers
    listed above. Any opinions, findings and conclusions or
    recommendations expressed in this material are those of the author
    and do not necessarily reflect the views of the National Science
    Foundation (NSF). \texttt{\{cbadescu,odonnell\}@cs.cmu.edu}}
}

\allowdisplaybreaks

\title{Lower bounds for testing complete positivity \\ and quantum separability}
\author{Costin B\u{a}descu\thanksNote \and Ryan O'Donnell\footnotemark[1]}
\date{\today}

\renewcommand*{\sD}{\mathsf{D}}
\renewcommand*{\sX}{\mathsf{Z}}
\DeclareMathOperator{\CPD}{CPD}
\DeclareMathOperator{\CP}{CP}
\DeclareMathOperator{\Uniform}{Unif}
\newcommand*{\Unif}{\Uniform_{d^2}}
\DeclarePairedDelimiterX\diverg[2]{(}{)}{#1 \mathrel{}\mathclose{}\delimsize\|\mathopen{}\mathrel{} #2}
\newcommand{\KL}[2]{\mathrm{KL}\diverg{#1}{#2}}

\begin{document}
\maketitle

\begin{abstract}
  In this work we are interested the problem of testing quantum entanglement.
  More specifically, we study the \emph{separability} problem in quantum
  property testing, where one is given $n$ copies of an unknown mixed
  quantum state $\varrho$ on $\CC^d \tensor \CC^d$, and one wants to
  test whether $\varrho$ is separable or $\eps$-far from all separable
  states in trace distance. We prove that $n = \Omega(d^2/\eps^2)$
  copies are necessary to test separability, assuming $\eps$ is not too
  small, viz.\ $\eps = \Omega(1/\sqrt{d})$.

  We also study completely positive distributions on the grid $[d]
  \times [d]$, as a classical analogue of separable states.  We
  analogously prove that $\Omega(d/\eps^2)$ samples from an unknown
  distribution $p$ are necessary to decide whether $p$ is completely
  positive or $\eps$-far from all completely positive distributions in
  total variation distance.  Towards showing that the true complexity
  may in fact be higher, we also show that learning an unknown
  completely positive distribution on $[d] \times [d]$ requires
  $\Omega(d^2/\eps^2)$ samples.
\end{abstract}

\section{Introduction}

A bipartite quantum state $\varrho$ on $\CC^d \tensor \CC^d$ is said to be
\emph{separable} if it can be written as a convex combination of product
states, meaning states of the form $\rho_1 \tensor \rho_2$ where
$\rho_1$ and $\rho_2$ are quantum states on $\CC^d$. Separable quantum
states are precisely those states which do not exhibit any form of
quantum entanglement. These are the only states that can be prepared by
separated parties who can only share classical
information. Understanding the general structure and properties of the
set of separable states in higher dimensions is a difficult problem and
is the subject of much ongoing research. For instance, deciding whether a
given $d^2 \by d^2$ matrix represents a separable state on
$\CC^d \tensor \CC^d$ -- also known as the \emph{separability problem}
in the quantum literature -- is NP-hard~\cite{Gurvits:2004}. In this
work, we study the following property testing version of the
separability problem:
\begin{quote}
  Provided unrestricted measurement access to $n$ copies of an unknown
  quantum state $\varrho$ on $\CC^d \tensor \CC^d$, decide with high
  probability if $\varrho$ is separable or $\eps$-far from all separable
  states in trace distance.
\end{quote}
The ultimate goal is to determine the number of copies of $\varrho$ that is
necessary and sufficient to solve the problem, up to constant factors,
as a function of $d$ and $\eps$.

By estimating (i.e., fully learning) $\varrho$ using recent algorithms
for quantum state tomography~\cite{Haah:2016, O'Donnell:2016} and
checking if the estimate is sufficiently close to a separable state,
this problem can be solved using $O(d^4/\eps^2)$ copies of $\varrho$. In
this paper, we prove a lower bound, showing that $\Omega(d^2/\eps^2)$
copies of $\varrho$ are necessary when $\eps = \Omega(1/\sqrt{d})$; this
reaches a lower bound of $\Omega(d^3)$ for $\eps = \Theta(1/\sqrt{d})$.
Closing the gap between the known bounds seems like a difficult problem,
and we have no particularly strong feeling about whether the tight bound
is the upper bound, the lower bound, or something in between.  (Indeed,
at least one paper~\cite{Aubrun:2017b} contains some evidence that
$\widetilde{\Theta}(d^3)$ might be the true complexity for
constant~$\eps$.)

Given the difficulty of closing the gap, we have sought a classical
analogue of the separability testing problem to try as a first step.
Analogies between quantum states and classical probability distributions
have proven to be a helpful source of inspiration throughout quantum
theory. Unfortunately, entanglement is understood to be a purely quantum
phenomenon; every finitely-supported discrete distribution can be
expressed as a convex combination of product point distributions, so
there are no ``entangled'' distributions. But motivated by the
characterization of separable quantum states using symmetric extensions
and the quantum de Finetti theorem~\cite{Doherty:2004}, we propose as a
kind of analogue the study of mixtures of i.i.d.\ bivariate
distributions, which arise in the classical de Finetti theorem. Doherty
et al.~\cite{Doherty:2004} used the quantum de Finetti theorem to show
that a quantum state $\varrho$ on $\CC^d \tensor \CC^d$ is separable
(i.e.\ a mixture of product states) if and only if $\varrho$ has a
symmetric extension to $\CC^d \tensor (\CC^d)^{\tensor k}$ for any
positive integer $k$. Somewhat analogously, the classical de Finetti
theorem states that a sequence of real random variables is a mixture of
i.i.d.\ sequences of random variables if and only if it is
exchangeable~\cite{Diaconis:1977}.

We call distributions which are mixtures of i.i.d.\ bivariate
distributions \emph{completely positive}, due to their connection with
completely positive matrices. We show that, given sample access to an
unknown distribution $p$ over $[d] \times [d]$, at least $\Omega(d/\eps^2)$
samples are necessary to decide with high probability if $p$ is
completely positive or $\eps$-far from all completely positive
distributions in total variation distance. Our proof is a
generalization of Paninski's lower bound for testing if a distribution
is uniform~\cite{Paninski:2008}.

Regarding upper bounds, one can again get a trivial upper bound of $O(d^2/\eps^2)$ samples for testing complete positivity, simply by fully estimating~$p$ to $\eps$-accuracy in total variation distance.  We again do not know how to close the gap between $\Omega(d/\eps^2)$ and $O(d^2/\eps^2)$, but we present evidence that the upper bound may be the true complexity.  Specifically, a common strategy for trying to test a family $D$ of distributions is to solve the problem of \emph{learning} an unknown distribution promised to be in~$D$.  We show that learning a completely positive distribution on $[d] \times [d]$ to accuracy~$\eps$ requires $\Omega(d^2/\eps^2)$ samples.  On the other hand, we are not able to show an analogous improved lower bound for learning separable quantum states.

\subsection{Previous work}

The property testing version of the separability problem, as defined
above, appears in~\cite{Montanaro:2016}, where a lower bound of
$\Omega(d^2)$ is proven for constant $\eps$. As
in~\cite{Montanaro:2016}, our proof also reduces the problem of testing
if a state is separable to the problem of testing if a state is the
maximally mixed state. However, we do not pass through the notion of
entanglement of formation, as~\cite{Montanaro:2016} does, and instead
rely on results about the convex structure of the set of separable
states. This approach yields a more direct proof that certain random
states are w.h.p.\ far from separable, which allows us to take advantage
of a lower bound from~\cite{O'Donnell:2015} (see~\Cref{thm:from-qst}).

We believe that the separability testing problem has seen further study, but
that there has been a lack of results due to its difficulty.  There \emph{is} a very extensive literature on the subject of entanglement detection (see e.g.~\cite{Guehne:2009, Horodecki:2009}), which is concerned with establishing different criteria for detecting or verifying entanglement. However, it
is not obvious how these results can be applied in the property testing
setting. In particular, few of these criteria are specifically concerned
with states that are far from separable in trace distance and many only
apply to certain restricted classes of quantum states.

As regards our classical analogue -- testing if a bipartite distribution is completely positive (mixture of i.i.d.)\ --  we are not aware of previous work in the literature.  The proof of our $\Omega(d/\eps^2)$ lower bound is inspired by, and generalizes, Paninski's lower bound for testing if a distribution is uniform~\cite{Paninski:2008}. The proof of our tight $\Omega(d^2/\eps^2)$ lower bound for learning completely positive distributions uses the Fano inequality method.

\subsection{Outline}

In~\Cref{sec:prelim} we cover background material on completely positive
distributions, quantum states and separability, and the property testing
framework that our results are concerned with. In~\Cref{sec:test-cp}, we
prove that testing if a distribution $p$ on $[d] \times [d]$ is
completely positive or $\eps$-far from all completely positive
distributions in total variation distance requires $\Omega(d/\eps^2)$
samples from $p$; we also prove that learning completely positive distributions requires $\Omega(d^2/\eps^2)$ samples. Finally, in~\Cref{sec:test-sep}, we show that testing
if a quantum state $\varrho$ on $\CC^d \tensor \CC^d$ is separable or
$\eps$-far from all separable states in trace distance requires
$\Omega(d^2/\eps^2)$ copies of $\varrho$ when
$\eps = \Omega(1/\sqrt{d})$.

\section{Preliminaries}\label{sec:prelim}

This section covers the mathematical background and notation used in the
rest of the paper.

\subsection{Completely positive distributions}

There is a well-developed theory of completely positive and copositive
matrices (see e.g.~\cite[Chapter 7]{Gaertner:2012}). In this section, we review
some known material.

Let $d$ be a positive integer. We consider distributions over the grid
$[d]^2 = [d] \times [d] = \set{(1, 1), (1, 2), \dotsc, (d, d)}$ which we
represent as matrices $A \in \RR^{d \by d}$ with $A_{ij}$ being the
probability of sampling $(i, j)$.

\begin{example}
  If $p \in \RR^d$ is a distribution on $[d] = \set{1, \dotsc, d}$
  represented as a column vector, then $p p^\tp$ is the natural i.i.d.\
  product probability distribution on $[d] \times [d]$ derived from $p$,
  with $p_i p_j$ being the probability of sampling $(i, j)$.
\end{example}

\begin{definition}\label{def:completely-positive}
  A matrix $A \in \RR^{d \by d}$ is \emph{completely positive} (CP) if
  there exist vectors $v_1, \dotsc, v_k \in \RR^d_{\ge 0}$ with
  nonnegative entries such that $A$ can be expressed as a convex
  combination of their projections $v_1 v_1^\tp, \dotsc, v_k v_k^\tp$,
  viz.
  \begin{align}\label{eqn:cp-def}
    A
    &= \sum_{i=1}^k c_i v_i v_i^\tp
  \end{align}
  for some nonnegative real numbers $c_1, \dotsc, c_k \in \RR$ with
  $c_1 + \dotsb + c_k = 1$.

  A distribution on $[d]^2$ represented as a matrix $A$ is
  \emph{completely positive} if $A$ is a CP matrix.
\end{definition}

\begin{remark}\label{prop:cp-cvx-comb}
  For a CP distribution $A$, the vectors $v_i$ in~\Cref{eqn:cp-def} may
  be taken to be probability distributions, since one can replace $v_i$
  by $v_i / \norm{v_i}_1$ and $c_i$ by $c_i \norm{v_i}_1^2$. Thus, CP
  distributions are precisely the mixtures of i.i.d.\ distributions.
\end{remark}

It follows immediately from~\Cref{def:completely-positive} that a CP
matrix $A$ satisfies three basic properties:
\begin{enumerate}[label=(\roman*)]
\item $A$ is symmetric ($A^\tp = A$),
\item $A_{ij} \ge 0$ for all $i, j \in [d]$, and
\item $A$ is positive semidefinite (PSD), denoted $A \ge 0$.
\end{enumerate}
A matrix satisfying these three properties is called \emph{doubly
  nonnegative}. However, if $d \ge 5$, then there exist doubly
nonnegative matrices which are not completely
positive~\cite{Maxfield:1962}.

\begin{example}
  Let $J$ denote the $d \by d$ matrix with $J_{ij} = 1$ for all
  $i, j \in [d]$ and let $\Unif = J/d^2$ denote the uniform
  distribution on $[d]^2$. Since
  $\Unif = (\frac{1}{d}, \dotsc, \frac{1}{d}) (\frac{1}{d}, \dotsc,
  \frac{1}{d})^\tp$, the uniform distribution on $[d]^2$ is completely
  positive.
\end{example}

Let $\CP_d$ denote the set of completely positive $d \by d$ matrices and
let $\CPD_d$ denote its subset of completely positive distributions on
$[d]^2$. It is well known that $\CP_d$ is a cone and that its dual cone
consists of \emph{copositive} matrices, i.e.\ matrices $M$ such that
$x^\tp M x \ge 0$ for all nonnegative vectors $x \in \RR^d_{\ge
  0}$. Thus, by cone duality, if $B \not\in \CP_d$ is a non-CP matrix,
then there exists a copositive matrix $W$ such that $\tr(AW) \ge 0$ for
all $A \in \CP_d$ and $\tr(BW) < 0$. This result yields witnesses
certifying nonmembership in $\CPD_d$. However, its usefulness is limited
by the fact that it provides no quantitative information about how far a
nonmember $A$ is from the set $\CPD_d$.

In what follows, we interpret distributions on $[d]^2$ as weighted
directed graphs with self-loops and obtain a sufficient condition for a
distribution to be $\eps$-far in total variation distance from $\CPD_d$
in terms of the maximum value of a cut in the corresponding graph.

We interpret a distribution $A$ on $[d]^2$ as a weighted directed graph
$G$ with vertices $V(G) = [d]$ and edges
\begin{align*}
  E(G)
  &= \set{(i, j) \in [d]^2 \mid A_{ij} > 0}.
\end{align*}
A \emph{cut} $x \in \set{\pm 1}^d$ in $G$ is a bipartition of the
vertices $V(G) = E_1 \cup E_2$ with $E_1 = \set{i \in [d] \mid x_i < 0}$
and $E_2 = \set{i \in [d] \mid x_i > 0}$. The total weight of edges cut
by this bipartition is
\begin{align*}
  \sum_{(i, j) \in [d]^2} \half[1 - x_i x_j] A_{ij}
  &= \E_{(\bs i, \bs j) \sim A} \half[1 - x_{\bs{i}} x_{\bs{j}}]
    = \half - \half \E_{(\bs i, \bs j) \sim A} x_{\bs{i}} x_{\bs{j}}
    = \half - \half x^\tp A x.
\end{align*}
In particular, if $A = p p^\tp$ with $p \in \RR^d$, then
\begin{align*}
  x^\tp A x
  &= x^\tp p p^\tp x
    = (x^\tp p)^2
    \ge 0.
\end{align*}
By~\Cref{prop:cp-cvx-comb}, a CP distribution is a convex combination of
matrices of the form $p p^\tp$. Thus, the following holds:
\begin{proposition}
  If $A$ is a CP distribution, then the total weight of a cut in the
  graph represented by $A$ is at most $\half$.
\end{proposition}
This fact allows us to prove the following result which gives a
sufficient condition for a distribution to be $\eps$-far from all CP
distributions in $\ell^1$ distance.  (The matrix norms in the following
are entrywise.)

\begin{proposition}\label{prop:eps-far-cp}
  Let $A$ be a distribution on $[d]^2$. If there exists a cut
  $x \in \set{\pm 1}^d$ with $x^\tp A x \le - \eps$, then
  $\norm{B - A}_1 \ge \eps$ for all $B \in \CPD_d$.
\end{proposition}
\begin{proof}
  Let $B \in \CPD_d$ be arbitrary. By H\"{o}lder's inequality, for all
  $U \in \RR^{d \by d}$ with $\norm{U}_\infty = 1$,
  \begin{align*}
    \norm{B - A}_1
    &\ge \tr(U^\tp (B - A))
      = \tr(U^\tp B) - \tr(U^\tp A).
  \end{align*}
  Let $U = x x^\tp$. Since $x^\tp B x \ge 0$ and
  $\tr(U^\tp A) = x^\tp A x \le - \eps$,
  \begin{align*}
    \norm{B - A}_1
    &\ge x^\tp B x - x^\tp A x
      \ge \eps. \qedhere
  \end{align*}
\end{proof}

\subsection{Quantum states and separability}

This section serves as a brief introduction to quantum states and
separability. For a more comprehensive introduction, see
e.g.~\cite{Watrous:2018}.

We work over $\CC$ and use bra--ket notation to denote vectors in
$\CC^d$, viz.\ for all vectors $x, y \in \CC^d$ and matrices
$A \in \CC^{d \by d}$, $\ket{x} = x$, $\bra{x} = x^\dag = \bar{x}^\tp$,
$\bra{x \tensor y} = \bra{x} \tensor \bra{y}$,
$\ket{x \tensor y} = \ket{x} \tensor \ket{y}$,
$\braket{x}{y} = x^\dag y$, $\ketbra{x}{y} = x y^\dag$, and
$\bra{x} A \ket{y} = x^\dag A y$.

\begin{definition}
  A \emph{quantum state} $\rho$ on $\CC^d$ is a positive semidefinite
  matrix $\rho \in \CC^{d \by d}$ with $\tr(\rho) = 1$. A
  \emph{measurement} is a set $\set{E_1, \dotsc, E_k}$ of positive
  semidefinite matrices on $\CC^d$ with $E_1 + \dotsb + E_k = \unit$,
  where $\unit$ denotes the identity matrix.
\end{definition}

Let $\rho$ and $\set{E_1, \dotsc, E_k}$ be as in the definition above
and let $p_i = \tr(\rho E_i)$ for $i = 1, \dotsc, k$. Since $\rho$ and
the $E_i$ are PSD, $p_i \ge 0$ for all $i = 1, \dotsc, k$, and
\begin{align*}
  p_1 + \dotsb + p_k
  &= \tr(\rho E_1) + \dotsb + \tr(\rho E_k)
    = \tr(\rho (E_1 + \dotsb + E_k))
    = \tr(\rho)
    = 1.
\end{align*}
Hence, $(p_1, \dotsc, p_k)$ is a distribution on $[k]$. Applying the
measurement $\set{E_1, \dotsc, E_k}$ to the quantum state $\rho$ yields
outcome $i \in [k]$ with probability $p_i = \tr(\rho E_i)$.

\begin{example}
  $\frac{\unit}{d}$ is a quantum state on $\CC^d$ called the
  \emph{maximally mixed state}; it is analogous to the uniform
  distribution on $[d]$.
\end{example}

\begin{definition}
  A state of the form $\rho = \ketbra{x}{x}$ for some $x \in \CC^d$ is
  called a \emph{pure} state.
\end{definition}

Given quantum states $\rho$ and $\sigma$ on $\CC^d$, the tensor product
$\rho \tensor \sigma$ is a quantum state on $\CC^d \tensor \CC^d$. If
$\rho$ and $\sigma$ represent the individual states of two isolated
particles, then $\rho \tensor \sigma$ is the state of the physical
system comprising both particles. Thus, the system composed of $n$
identical copies of the state $\rho$ is represented as the state
$\rho^{\tensor n}$ on $(\CC^d)^{\tensor n}$.

\begin{definition}
  A quantum state $\varrho$ on $\CC^d \tensor \CC^d$ is \emph{separable}
  if $\varrho$ can be expressed as a convex combination of product
  states, viz.\
  \begin{align*}
    \varrho
    &= \sum_{i=1}^k c_i \rho_i \tensor \sigma_i,
  \end{align*}
  where $\rho_i$ and $\sigma_i$ are states on $\CC^d$ for
  $i = 1, \dotsc, k$ and $c_1, \dotsc, c_k \in \RR_{\ge 0}$ satisfy
  $c_1 + \dotsc + c_k = 1$. Thus, the physical system represented by
  $\varrho$ may be regarded as being in the state
  $\rho_i \tensor \sigma_i$ with probability $c_i$.

  A state that is not separable is called \emph{entangled}.
\end{definition}

\begin{example}
  Since $\frac{\unit}{d^2} = \frac{\unit}{d} \tensor \frac{\unit}{d}$,
  the maximally mixed state is separable.
\end{example}

\begin{definition}
  Let $\Sep$ denote the set of separable states on $\CC^d \tensor \CC^d$
  and let $\Sep_{\pm}$ denote its \emph{cylindrical symmetrization}
  (cf.~\cite[p.~81]{Aubrun:2017}), viz.\  $\Sep_{\pm} = \conv(\Sep \cup (- \Sep))$, where $\conv(E)$ denotes the convex hull of the set $E$.
\end{definition}

Similar to the duality between completely positive and copositive
matrices, the set $\Sep$ generates a cone of separable operators whose
dual is the cone of \emph{block-positive} operators (see
e.g.~\cite{Aubrun:2017}). A block-positive operator acts as an
entanglement witness certifying that a given quantum state is not
separable. Thus,~\Cref{prop:holder-distance} in \Cref{sec:test-sep} is comparable
to~\Cref{prop:eps-far-cp} in that it describes witnesses certifying that
a quantum state is not just entangled but actually $\eps$-far from all
separable states in trace distance.

\subsection{The property testing framework}

In the property testing model, we have a set $\cO$ of objects and also a
distance function \mbox{$\dist : \cO \times \cO \to \RR$}. A \emph{property}
$\cP$ is a subset of $\cO$ and the distance between an object
$x \in \cO$ and the property $\cP$ is defined by
$\dist(x, \cP) = \inf_{y \in \cP} \dist(x, y)$. An algorithm $\cT$ is
said to test $\cP$ if, given some type of access to $x \in \cO$ (e.g.\
independent samples or identical copies), $\cT$ accepts $x$ w.h.p.\ when $x \in
\cP$ and $\cT$ rejects $x$ w.h.p.\ when $\dist(x, \cP) \ge \eps$.

In~\Cref{sec:test-cp}, $\cO$ is the set of distributions on
$[d] \times [d]$, $\dist$ is the total variation distance, and
$\cP = \CPD_d \subset \cO$ is the set of CP distributions. Given samples
$\bs{x}_1, \dotsc, \bs{x}_n$ from a distribution $p$ on $[d]^2$, a
testing algorithm $\cT$ for $\CPD_d$ satisfies
\begin{align*}
  p \in \CPD_d &\implies \Pr[\cT(\bs{x}_1, \dotsc, \bs{x}_n)\ \text{accepts}] \ge
  \frac{2}{3}, \\
  p \text{ $\eps$-far from } \CPD_d &\implies \Pr[\cT(\bs{x}_1, \dotsc, \bs{x}_n)\ \text{accepts}] \le
  \frac{1}{3}.
\end{align*}

In~\Cref{sec:test-sep}, $\cO$ is the set of quantum states on
$\CC^d \tensor \CC^d$,
$\dist(\varrho, \sigma) = \half \norm{\varrho - \sigma}_1$ is the trace
distance between quantum states, and $\cP = \Sep$ is the set of
separable states on $\CC^d \tensor \CC^d$. Given measurement access to
$n$ copies $\varrho^{\tensor n}$ of a state
$\varrho \in \CC^d \tensor \CC^d$, a testing algorithm for $\Sep$ is a
two-outcome measurement $\set{E_0, E_1}$ on $(\CC^d)^{\tensor n}$
satisfying:
\begin{align*}
  \varrho \in \Sep &\implies \tr(E_1 \varrho^{\tensor n}) = 1] \ge
  \frac{2}{3}, \\
  \varrho \ \text{$\eps$-far from} \ \Sep &\implies \tr(E_1 \varrho^{\tensor n}) = 1] \le \frac{1}{3}.
\end{align*}

\section{Testing complete positivity}\label{sec:test-cp}

Let $d$ be a positive integer. If $d$ is odd, we can reduce to the case
of $d - 1$ by using distributions that don't involve outcome $d \in [d]$, and the asymptotics of
$\Omega(d/\eps^2)$ remain unchanged. Hence we may assume, without loss
of generality, that $d$ is even.

We begin by defining a family of distributions on $[d]^2$ which are
$\eps$-far from $\CPD_d$. Let $S \subset [d]$ be a subset of size
$\abs{S} = \half[d]$. Thus, $\abs{S^\co} = \half[d]$ and
\begin{align*}
  \abs{S \times S^\co \cup S^\co \times S}
  &= \abs{S \times S^\co} + \abs{S^\co \times S}
    = \half[d^2].
\end{align*}
Let $\phi_S : [d]^2 \to \RR$ be the function defined by
\begin{align*}
  \phi_S(x) =
  \begin{cases}
    1 + \eps, & x \in S \times S^\co \cup S^\co \times S \\
    1 - \eps, & \text{otherwise}.
  \end{cases}
\end{align*}
Hence,
\begin{align*}
  \avg_{x \in [d]^2} \phi_S(x)
  &= \frac{1}{d^2} \paren*{\half[d^2] (1 + \eps) + \half[d^2] (1 -
    \eps)}
    = 1.
\end{align*}
So we may think of $\phi_S$ as a density function with respect to the
uniform distribution on $[d]^2$.

Let $x \in \set{\pm 1}^d$ be defined as follows: for all $i \in [d]$, if
$i \in S$, then $x_i = 1$, otherwise $x_i = -1$. Let $A^S$ be the matrix
defined by $A^S_{ij} = \phi_S((i, j)) / d^2$. Thus, $A^S$ is a symmetric
distribution on $[d]^2$ and $x$ is a cut. The total weight of this cut
is
\begin{align*}
  \half[d^2] \cdot \frac{1 + \eps}{d^2}
  = \half + \half[\eps].
\end{align*}
Therefore, for every subset $S \subset [d]$, the distribution $A^S$ is
\emph{not} completely positive. Moreover, $x^\tp A^S x = - \eps$, so,
by~\Cref{prop:eps-far-cp},
\begin{align*}
  \norm{A^S - B}_1 \ge \eps
\end{align*}
for every CP distribution $B$ (where the matrix norm is entry-wise). In other words, for every subset
$S \subset [d]$ with $\abs{S} = \half[d]$, $A^S$ is a distribution on
$[d]^2$ which is $\eps$-far in $\ell^1$ distance from every CP
distribution on $[d]^2$.

Fix $\Omega = [d]^2$ and let $\phi : \Omega^n \to \RR$ denote the
function defined by
\begin{align*}
  \phi(x)
  &= \avg_{\substack{S \subset [d] \\ \abs{S} = d/2}} \phi_S(x_1) \dotsm \phi_S(x_n).
\end{align*}
Let $\cD_n$ denote the distribution on $\Omega^n$ defined by the density
$\phi$ and let $d_{\chi^2}(\place, \place)$ denote the $\chi^2$-distance
between probability distributions, i.e.\ for distributions $\cP$ and
$\cQ$ on $\Omega$,
\begin{align*}
  d_{\chi^2}(\cP, \cQ)
    &= \E_{\bs{x} \sim \cQ} \bracket*{\paren*{\frac{\cP(\bs x)}{\cQ(\bs x)} - 1}^2}.
\end{align*}

The following proposition will be shown to imply our lower bound:
\begin{proposition}\label{prop:chi-squared-lower-bound}
  If $d_{\chi^2}(\cD_n, \Unif^{\tensor n}) \ge \frac{1}{3}$, then
  $n = \Omega(d / \eps^2)$.
\end{proposition}
\begin{proof}
  Let $\cH$ denote the uniform distribution over subsets $S
  \subset [d]$ with $\abs{S} = d/2$. Thus,
  \begin{align*}
    d_{\chi^2}(\cD_n, \Unif^{\tensor n})
    &= \paren*{ \sum_{x \in \Omega^n} \frac{\cD_n(x)^2}{\Unif^{\tensor
      n}(x)}} - 1 \\
    &= \paren*{ \sum_{x \in \Omega^n} \frac{\phi(x)^2}{d^{2n}}} - 1 \\
    &= \E_{{\bs x} \sim \Unif^{\tensor n}} \phi(\bs x)^2 - 1 \\
    &= \E_{{\bs x} \sim \Unif^{\tensor n}} \bracket*{\paren*{\E_{\bs{S} \sim \cH}
      \phi_{\bs S}(\bs{x}_1) \dotsm
      \phi_{\bs S}(\bs{x}_n)}^2} - 1 \\
    &= \E_{{\bs x} \sim \Unif^{\tensor n}} \bracket*{\E_{\bs{S}, \bs{S'} \sim \cH}
      \phi_{\bs S}(\bs{x}_1) \dotsm \phi_{\bs S}(\bs{x}_n)
      \phi_{\bs{S'}}(\bs{x}_1) \dotsm \phi_{\bs{S'}}(\bs{x}_n)} - 1 \\
    &= \E_{\bs{S}, \bs{S'} \sim \cH} \E_{{\bs x} \sim \Unif^{\tensor n}}
      \phi_{\bs S}(\bs{x}_1) \dotsm \phi_{\bs S}(\bs{x}_n)
      \phi_{\bs{S'}}(\bs{x}_1) \dotsm \phi_{\bs{S'}}(\bs{x}_n) - 1 \\
    &= \E_{\bs{S}, \bs{S'} \sim \cH} \bracket*{\paren*{\E_{{\bs x} \sim \Unif}
      \phi_{\bs S}(\bs{x}) \phi_{\bs{S'}}(\bs{x})}^n} - 1.
  \end{align*}
  For a subset $E \subset [d]$, let $\chi_E$ be the $\pm 1$-valued
  indicator function defined by $\chi_E(x) = 1$ if $x \in E$ and
  $\chi_E(x) = -1$ otherwise. Note that
  $\phi_E(x) = 1 - \chi_{E}(x_1) \chi_{E}(x_2) \eps$ for all
  $x \in \Omega$. Hence,
  \begin{align*}
    \phi_{\bs S}(\bs{x}) \phi_{\bs S'}(\bs{x})
    &= 1 - (\chi_{\bs S}(\bs{x}_1) \chi_{\bs S}(\bs{x}_2) + \chi_{\bs S'}(\bs{x}_1)
      \chi_{\bs S'}(\bs{x}_2))\eps
      + \chi_{\bs S}(\bs{x}_1) \chi_{\bs S}(\bs{x}_2) \chi_{\bs S'}(\bs{x}_1) \chi_{\bs
      S'}(\bs{x}_2) \eps^2.
  \end{align*}
  For a fixed outcome of $\bs{S}$ and $\bs{x}$ uniformly random,
  $\chi_{\bs S}(\bs{x}_1)$ and $\chi_{\bs S}(\bs{x}_2)$ are independent
  uniform $\pm 1$-valued bits. So, in expectation, the terms involving
  just $\eps$ in the expression above drop out. Moreover, $\chi_{\bs
    S}(\bs{x}_1) \chi_{\bs S'}(\bs{x}_1)$ and $\chi_{\bs S}(\bs{x}_2)
  \chi_{\bs S'}(\bs{x}_2)$ are independent. Hence,
  \begin{align*}
    \E_{{\bs x} \sim \Unif}
      \phi_{\bs S}(\bs{x}) \phi_{\bs{S'}}(\bs{x})
    &= 1 - \eps^2 \cdot \paren*{\E_{{\bs x} \sim \Unif} \chi_{\bs S}(\bs{x}_1) \chi_{\bs S'}(\bs{x}_1)}^2
  \end{align*}
  Let $\bs{r} = \abs{\bs{S} \cap \bs{S'}}$, where
  $\bs{S}, \bs{S'} \sim \cH$, and let $\bs{\delta}$ denote the mean of
  $\chi_{\bs S}(\bs{x}_1) \chi_{\bs S'}(\bs{x}_1)$ appearing above. It
  is easy to check that $\bs{\delta} = 4 \bs{r} / d - 1$. Thus,
\begin{align*}
  d_{\chi^2}(\cD_n, \Unif^{\tensor n})
  &\le \E_{\bs{S}, \bs{S'} \sim \cH} \bracket*{\paren*{1 + \eps^2
    \bs{\delta}^2}^n} - 1 \\
  &\le \E_{\bs{S}, \bs{S'} \sim \cH} \bracket*{\exp(\eps^2
    \bs{\delta}^2)^n} - 1 \\
  &= \E_{\bs{S}, \bs{S'} \sim \cH} \bracket*{\exp(n \eps^2
    \bs{\delta}^2)} - 1.
\end{align*}
Since $\exp(n \eps^2 \bs{\delta}^2) - 1 \ge 0$,
\begin{align*}
  \E_{\bs{S}, \bs{S'} \sim \cH} \bracket*{\exp(n \eps^2
  \bs{\delta})} - 1
  &= \int_0^\infty \Pr_{\bs{S}, \bs{S'} \sim \cH} \bracket*{\exp(n
    \eps^2 \bs{\delta}^2) - 1 \ge t} dt.
\end{align*}
Since $\exp(n \eps^2 \bs{\delta}^2) - 1 \ge t$ is equivalent to
\begin{align*}
  \bs{r}
  &\ge \frac{d}{4} + \frac{d}{4} \cdot \paren*{\frac{\log(1 + t)}{n
    \eps^2}}^{\half}
\end{align*}
it follows that
\begin{align*}
  d_{\chi^2}(\cD_n, \Unif^{\tensor n})
  &\le \int_0^\infty \Pr_{\bs{S}, \bs{S'} \sim \cH} \bracket*{\bs{r} \ge
    \frac{d}{4} + \frac{d}{4} \sqrt{f(t)}} dt,
\end{align*}
where $f(t) = \log(1 + t) / n \eps^2$.

Since $\bs{r} = \abs{\bs{S} \cap \bs{S'}}$ is invariant under
permutations of $[d]$, it follows that $\bs{r}$ is distributed according
to the hypergeometric distribution with $d/2$ draws from a set of $d$
elements with $d/2$ successes. If $\bs{X}$ is a random variable
distributed according to the hypergeometric distribution with $m$ draws
from a set of $N$ elements with $k$ successes, then (see
e.g.~\cite{Skala:2013})
\begin{align*}
  \Pr \bracket*{\frac{\bs{X}}{m} \ge \frac{k}{N} + s} \le \exp(-2s^2m).
\end{align*}
Hence,
\begin{align*}
  \Pr_{\bs{S}, \bs{S}' \sim \cH} \bracket*{\bs{r} \cdot \frac{2}{d}
  \ge \half + t}
  &= \Pr_{\bs{S}, \bs{S}' \sim \cH} \bracket*{\bs{r}
    \ge \frac{d}{4} + \frac{dt}{2}}
    \le \exp(-dt^2),
\end{align*}
whence,
\begin{align*}
  \Pr_{\bs{S}, \bs{S}' \sim \cH} \bracket*{\bs{r}
  \ge \frac{d}{4}(\sqrt{f(t)} + 1)}
  &\le \exp(-df(t)/4).
\end{align*}
Therefore,
\begin{align*}
  d_{\chi^2}(\cD_n, \Unif^{\tensor n})
  &\le \int_0^\infty \exp(-df(t)/4) dt \\
  &= \int_0^\infty \exp \paren*{- \frac{d}{4 n \eps^2} \log (1 + t)} dt
  \\
  &= \int_0^\infty \paren*{\frac{1}{1 + t}}^c dt, \\
  &= \frac{1}{c - 1},
\end{align*}
where $c = d / 4 n \eps^2$. Since $d_{\chi^2}(\cD_n, \Unif^{\tensor n})
\ge 1/3$, it follows that $c \le 4$, so $n \ge d / 16
\eps^2$. Therefore, $n = \Omega(d / \eps^2)$, as needed.
\end{proof}

Let $d_{\TV}(\place, \place)$ denote the total variation distance
between probability distributions. Let $p \in \CPD_d$ and let $q$ be a
distribution $\eps$-far from $\CPD_d$.

A testing algorithm $f : ([d]^2)^n \to \Bin$ for complete positivity
determines a probability event $E \subset ([d]^2)^n$ satisfying
$p^{\tensor n}(E) \ge 2/3$ and $q^{\tensor n}(E) \le 1/3$. Hence,
$\Unif^{\tensor n}(E) \ge 2/3$ and, since $\cD_n$ is supported on
distributions $\eps$-far from $\CPD_d$, $\cD_n(E) \le 1/3$. Therefore,
$d_{\TV}(\cD_n, \Unif^{\tensor n}) \ge 1/3$ and the following
corollary establishes the lower bound:

\begin{corollary}
  If $d_{\TV}(\cD_n, \Unif^{\tensor n}) \ge 1/3$, then
  $n = \Omega(d/\eps^2)$.
\end{corollary}
\begin{proof}
  For all distributions $\mu$ and $\nu$, $2 d_{\TV}(\mu, \nu)^2 \le
  d_{\chi^2}(\mu, \nu)$. Hence,
  \begin{align*}
    (d/4n\eps^2 - 1)^{-1}
    &\ge d_{\chi^2}(\cD_n, \Unif^{\tensor n})
      \ge 2d_{\TV}(\cD_n, \Unif^{\tensor n})^2
      \ge \frac{2}{9},
  \end{align*}
  where the first inequality is obtained in the proof
  of~\Cref{prop:chi-squared-lower-bound}. Therefore,
  $n = \Omega(d/\eps^2)$.
\end{proof}

\subsection{Hardness of learning completely positive distributions}

Given our $\Omega(d/\eps^2)$ lower bound for testing if a distribution $p$ on $[d] \times [d]$ is completely positive, it is natural to ask whether one can improve the trivial upper bound of $O(d^2/\eps^2)$. This trivial upper bound comes from fully learning $p$ to $\eps/2$ accuracy in trace distance and then checking if the estimate is $\eps/2$-close to~CP.  A natural strategy for improving the testing upper bound is the \emph{testing by learning} method; if we could show an $o(d^2/\eps^2)$ upper bound for \emph{learning} an unknown $p$ on $[d] \times [d]$ under the assumption that it is CP, this would yield an improved testing upper bound.  However, here we show that this is not possible.

\begin{theorem}                                     \label{thm:learning-CP-hard}
    Let $\mathcal{L}$ be an algorithm that gets as input a parameter $0 < \eps < \frac12$ as well as access to $n$~samples from a completely positive probability distribution $p$ on $[d] \times [d]$.  Suppose $\mathcal{L}$ is guaranteed to output a hypothesis distribution $\widehat{p}$ satisfying $d_{\TV}(p,\widehat{p}) \leq .1\eps$ with probability at least $\frac23$.  Then $n \geq \Omega(d^2/\eps^2)$.
\end{theorem}
\begin{proof}
    By the Fano inequality method, it suffices to construct completely positive distributions $p_1, \dots, p_N$ on $[d] \times [d]$ with: 
    \begin{equation} \label{eqn:CP-reqs}
        \text{(i) } N \geq 2^{\Omega(d^2)}; \quad \text{(ii) } d_{\TV}(p_i,p_j) > .1\eps \text{ for all $i \neq j$;} \quad  \text{(iii) } \KL{p_i}{p_j} \leq O(\eps^2) \text{ for all $i \neq j$,}
    \end{equation}
    where $\KL{\place}{\place}$ denotes Kullback--Leibler divergence.

    To get started on this, let $\mathcal{I}$ denote the set of $\binom{d}{2}$ ``upper-triangular'' pairs $\{(i,j) \in [d] \times [d] : i < j\}$; similar to before we may assume, without of loss of generality, that $d$ is a multiple of~$4$ and hence $|\mathcal{I}|$ is even.  Given $(i,j) \in \mathcal{I}$, let us define the probability distribution $q_{(i,j)}$ on $[d] \times [d]$ by:
    \[
            q_{(i,j)} \text{ has probability mass $\frac14$ on each of } (i,i),\ (i,j),\ (j,i),\ (j,j).
    \]
    Note that $q_{(i,j)}$ is a product distribution, namely the two-fold product of the distribution on $[d]$ that puts mass $\frac12$ on $i$ and mass $\frac12$ on $j$.  Thus any mixture of $q_{(i,j)}$'s is a completely positive distribution.
    
    Next, let $\mathcal{M} = (e_1, \dots, e_m)$ denote an arbitrary perfect matching on~$\mathcal{I}$, where  $m = \frac12 \binom{d}{2}$. If $e = ( (i_0,j_0), (i_1,j_1))$ is a typical matched pair from~$\mathcal{M}$, define the following two completely positive distributions:
    \[
        r_e^{(0)} = \frac{1+\eps}{2} q_{(i_0,j_0)} + \frac{1-\eps}{2} q_{(i_1,j_1)}, \qquad
        r_e^{(1)} = \frac{1-\eps}{2} q_{(i_0,j_0)} + \frac{1+\eps}{2} q_{(i_1,j_1)}.
    \]
    
    Next, given $\omega \in \{0,1\}^m$, define the following completely positive distribution:
    \[
        \widetilde{p}_\omega = \avg_{k=1\dots m} r_{e_k}^{(\omega_k)}.
    \]
    That is, $\widetilde{p}_\omega$ is an equal mixture of some half of the the $r_e^{(\cdot)}$ distributions, with $\omega$ selecting which of $r_e^{(0/1)}$ is taken for each $e \in \mathcal{M}$.
    
    Let us now consider the ``diagonal'' probability mass of the $\widetilde{p}_\omega$ distributions.  First, since every $q_{(i,j)}$ puts $\frac12$ of its probability mass on diagonal elements, the same is true of the $r_{e}^{(\cdot)}$ and $\widetilde{p}_\omega$ distributions.  Furthermore, it is easy see that $\widetilde{p}_\omega(i,i) = \frac{1+\delta^i_\omega}{2d}$ for numbers $\delta^i_\omega \in [-\epsilon, \epsilon]$ for all $i \in [d]$ and $\omega \in \{0,1\}^m$.  Thus for each $\omega \in \{0,1\}^m$ we may construct a ``correcting'' distribution $u_\omega$ on $[d] \times [d]$, completely supported on diagonal elements with $u_\omega(i,i) = \frac{1-\delta^i_\omega/2}{d}$, such that
    \[
        p_\omega \coloneqq \frac12 \widetilde{p}_\omega + \frac12 u_\omega
    \]
    has $p_\omega(i,i) = \frac{3}{4d}$ for all $i \in [d]$.  In other words, each $p_\omega$ has $\frac{3}{4}$ of its probability mass on the diagonal, and is uniform on the diagonal elements.  Furthermore, the distributions $p_\omega$ are still completely positive, since every diagonal distribution is completely positive.  
    
    Our final collection of CP distributions $p_1, \dots, p_N$ satisfying \Cref{eqn:CP-reqs} will be a subset of the~$p_\omega$'s.  Since any two $p_\omega$, $p_{\omega'}$ agree on the diagonal elements, we may calculate:
    \[
        d_{\TV}(p_\omega, p_{\omega'}) = \sum_{(i,j) \in \mathcal{I}} \abs{p_\omega(i,j) - p_{\omega'}(i,j)} = \frac12 \avg_{k=1\dots m} \begin{dcases} \begin{rcases} \frac{\eps}{2} & \text{if $\omega_{k} \neq \omega'_{k}$} \\ 0 & \text{if $\omega_{k} = \omega'_{k}$} \end{rcases} \end{dcases}  = \frac{\Delta(\omega,\omega')}{4m}\eps,
    \]
    where $\Delta(\place,\place)$ is Hamming distance.  One may similarly compute
    \[
        \KL{p_\omega}{p_{\omega'}} = \frac{\Delta(\omega,\omega')}{8m} \eps \ln\left(\frac{1+\eps}{1-\eps}\right) \leq \frac{3\Delta(\omega,\omega')}{8m} \eps^2,
    \]
    where the inequality used $\eps < \frac12$.  To complete the proof via \Cref{eqn:CP-reqs}, it remains to choose a subset of $N = 2^{\Omega(m)} = 2^{\Omega(d^2)}$ elements of $\{0,1\}^m$ such that any two distinct chosen $\omega, \omega'$ have $\Delta(\omega, \omega') \geq .4m$.  This is easily done via the probabilistic method, or an explicit error-correcting code.
\end{proof}

\section{Testing separability}\label{sec:test-sep}

Let $d$ be a positive integer. As in the previous section, we may
assume, without loss of generality, that $d$ is even.

Let $\cH = \CC^d \tensor \CC^d$, let $\U(\cH)$ denote the set of unitary
operators on $\cH$, and recall that $\Sep$ denotes the set of separable
states on $\cH$. For all operators $T$ on $\cH$, let $\norm{T}_p$ denote
the Schatten $p$-norm of $T$, viz.\
$\norm{T}_p = \paren*{\tr(\abs{T}^p)}^{\frac{1}{p}}$, where
$\abs{T} = \sqrt{T^\dag T}$ is the absolute value of the operator
$T$. Let $d_{\tr}(\varrho, \sigma) = \half \norm{\varrho - \sigma}_1$
denote the trace distance between quantum states $\varrho$ and $\sigma$.

We begin by defining a family of quantum states which are with high
probability $O(\eps)$-far from $\Sep$. For $0 \le \eps \le 1/2$, let
$\sD_\eps$ be the diagonal matrix on $\cH$ defined by
\begin{align*}
  \sD_\eps
  &= \diag \paren*{\frac{1+2\eps}{d^2}, \dotsc, \frac{1+2\eps}{d^2},
    \frac{1-2\eps}{d^2}, \dotsc, \frac{1-2\eps}{d^2}},
\end{align*}
where $\tr(\sD_\eps) = 1$, and let $\cD$ denote the family of all
quantum states on $\cH$ with the same spectrum as $\sD_\eps$, viz.\
$\cD = \set{U \sD_\eps U^\dag \mid U \in \U(\cH)}$.

Our lower bound will rely on the following theorem which follows
immediately from~\cite[Lemma 2.22 and Theorem 4.2]{O'Donnell:2015}:
\begin{theorem}\label{thm:from-qst}
  $\Omega(d^2/\eps^2)$ copies are necessary to test whether a quantum
  state $\varrho$ on $\cH$ is the maximally mixed state or $\varrho \in
  \cD$.
\end{theorem}

If $\bs{U}$ is a random unitary on $\cH$ distributed according to the
Haar measure, then $\bs{\varrho} = \bs{U} \sD_\eps \bs{U}^\dag$ is a random
element of $\cD$. This induced probability measure is invariant under
conjugation by a fixed unitary: for all $V \in \U(\cH)$,
$V \bs{\varrho} V^\dag$ has the same distribution as $\bs{\varrho}$. We want
to show the following:
\begin{lemma}\label{claim:eps-far-twothirds}
  There is a universal constant $C_0$ such that for all $C_0 / \sqrt{d}
  \le \eps \le 1/2$, the following holds when $\bs{\varrho} = \bs{U} \sD_\eps
  \bs{U}^\dag$ is a uniformly random state in $\cD$:
  \begin{align*}
    \Pr[\forall \sigma \in \Sep, \ \norm{\bs{\varrho} - \sigma}_1 \ge 2\eps]
    &\ge \frac{2}{3}.
  \end{align*}
\end{lemma}

As $\eps$ tends to zero, the elements of $\cD$ get closer to the
maximally mixed state and eventually become separable, by the
Gurvits--Barnum theorem~\cite{Gurvits:2002}. Indeed, if
$\eps \le 1/(2\sqrt{d^2 - 1})$, then $\cD \subset \Sep$. Hence, some
assumption on $\eps$ is necessary for~\Cref{claim:eps-far-twothirds} to
hold.

\Cref{claim:eps-far-twothirds} and~\Cref{thm:from-qst} easily imply the
desired lower bound:

\begin{theorem}
  Let $\varrho$ be a quantum state on $\CC^d \tensor \CC^d$ and let
  $\eps = \Omega(1/\sqrt{d})$. Testing if $\varrho$ is separable or
  $\eps$-far from $\Sep$ in trace distance requires $\Omega(d^2/\eps^2)$
  copies of $\varrho$.
\end{theorem}
\begin{proof}
  Let $\set{E_0, E_1}$ be a measurement corresponding to a separability
  testing algorithm using $n$ copies of $\varrho$. To apply the lower
  bound in~\Cref{thm:from-qst}, we use $\set{E_0, E_1}$ to define an
  algorithm that decides w.h.p.\ if a state $\varrho$ is equal to the
  maximally mixed state $\frac{\unit}{d^2}$ or $\varrho \in \cD$.

  Let $\varrho^{\tensor n}$ be given with either $\varrho \in \cD$ or
  $\varrho = \frac{\unit}{d^2}$. Note that, for all $\varrho \in \cD$,
  $d_{\tr}(\varrho, \frac{\unit}{d^2}) \ge \eps$ holds. Let $\bs{U}$ be
  a random unitary. If $\varrho$ is the maximally mixed state, then
  $V \varrho V^\dag = \varrho$ for all $V \in \U(\cH)$, so
  $(\bs{U} \varrho \bs{U}^\dag)^{\tensor n} = \varrho^{\tensor
    n}$. Otherwise, $\bs{U} \varrho \bs{U}^\dag$ is a random state in
  $\cD$.

  Applying the separability test $\set{E_0, E_1}$ to
  $\bs{U} \varrho \bs{U}^\dag$, we have that:
  \begin{enumerate}[label=(\roman*)]
  \item if $\bs{U} \varrho \bs{U}^\dag = \varrho = \frac{\unit}{d^2}$, then
    $\bs{U} \varrho \bs{U}^\dag$ is separable, so
    \begin{align*}
      \tr((\bs{U} \varrho \bs{U}^\dag)^{\tensor n} E_1)
      &= \tr(\varrho^{\tensor n} E_1)
        \ge \frac{2}{3}.
    \end{align*}

  \item if $\varrho \in \cD$, then the probability of error is
    \begin{align*}
      \E_{\bs{U}} \tr((\bs{U} \varrho \bs{U}^\dag)^{\tensor n} E_1)
      &\le \Pr[\bs{U} \varrho \bs{U}^\dag \ \text{is $\eps$-close to} \
        \Sep] + \Pr[\text{test fails} \mid \bs{U} \varrho \bs{U}^\dag \ \text{is
        $\eps$-far from $\Sep$}] \\
      &\le \frac{1}{3} + \frac{1}{3} \cdot \frac{2}{3}
        = \frac{5}{9},
    \end{align*}
    where the second inequality follows
    from~\Cref{claim:eps-far-twothirds}.
  \end{enumerate}
  Thus, using the separability test, we can distinguish w.h.p.\ between
  $\varrho = \frac{\unit}{d^2}$ and $\varrho \in \cD$ using $n$ copies
  of $\varrho$. Therefore, by~\Cref{thm:from-qst},
  $n = \Omega(d^2/\eps^2)$.
\end{proof}

It remains to show that~\Cref{claim:eps-far-twothirds} holds. Its proof
relies on two main facts: first, that $\Sep$ is approximated by a
polytope with $\exp(O(d)))$ vertices which are separable pure states;
and, second, that a random element of $\cD$ is $\eps$-far from a fixed
pure state except with probability $\exp(-O(d))$.

The first fact follows from the next lemma which is a rephrasing
of~\cite[Lemma 9.4]{Aubrun:2017}:
\begin{lemma}\label{lem:sep-approx}
  There exists a constant $C > 0$ such that, for every dimension $d$,
  there is a family $\cN$ of pure product states on $\cH$ (i.e.\ states
  of the form $\ketbra{x \tensor y}{x \tensor y}$ with $x, y \in \CC^d$)
  with $\abs{\cN} \le C^d$ satisfying
  \begin{align*}
    \conv(\cN \cup -\cN) \subset \Sep_\pm \subset 2 \conv(\cN \cup -\cN).
  \end{align*}
\end{lemma}

Now, we wish to upper bound the probability that a random element of
$\cD$ is $\eps$-far from a fixed pure state. The following result
provides a sufficient condition for a state $\sigma$ on $\cH$ to be
$\eps$-far from a state $\varrho \in \cD$:
\begin{proposition}\label{prop:holder-distance}
  Let $\varrho \in \cD$ be arbitrary and let
  $W = \dfrac{\unit}{d^2} - \varrho$. For all quantum states $\sigma$ on
  $\cH$, if $\tr(\sigma W) \ge - \eps \norm{W}_\infty$, then
  $\norm{\varrho - \sigma}_1 \ge \eps$.
\end{proposition}
\begin{proof}
  Note that
  \begin{align*}
    \tr(\varrho W)
    &= \frac{1}{d^2} - \tr(\varrho^2)
      = \frac{1}{d^2} - \frac{1 + 4 \eps^2}{d^2}
      = - \frac{4 \eps^2}{d^2}, \\
    \norm{W}_\infty
    &= \norm*{\frac{\unit}{d^2} - \sD_\eps}_\infty
      = \frac{2 \eps}{d^2}.
  \end{align*}
  By H\"{o}lder's inequality for matrices,
  $\tr((\sigma - \varrho)W) \le \norm{\sigma - \varrho}_1 \cdot
  \norm{W}_\infty$. Hence,
  \begin{align*}
    \norm{\sigma - \varrho}_1
    &\ge \frac{\tr(\sigma W) - \tr(\varrho W)}{\norm{W}_\infty}
      = 2 \eps + \frac{\tr(\sigma W)}{\norm{W}_\infty}. \qedhere
  \end{align*}
\end{proof}

When $\sigma = \ketbra{x}{x}$ with
$x \in \cH$ and $\varrho = U \sD_\eps U^\dag$, we have
\begin{align*}\label{eqn:trace-pure}
    \tr(\ketbra{x}{x} W)
    &= \bra{x} W \ket{x} \\
    &= \bra{x} \paren*{\frac{\unit}{d^2} - U \sD_\eps U^\dag}
      \ket{x} \\
    &= \bra{x} U \paren*{\frac{\unit}{d^2} - \sD_\eps}
      U^\dag \ket{x} \\
    &= \norm{W}_\infty \cdot \bra{x} U \sX U^\dag \ket{x}, \numbered
\end{align*}
where $\sX = \diag(-1, \dotsc, -1, 1, \dotsc, 1)$ is just
$\unit/d^2 - \sD_\eps$ divided by $\norm{W}_\infty$. Hence,
$\norm{\varrho - \ketbra{x}{x}}_1 \ge \eps$ holds if
$\bra{x} U \sX U^\dag \ket{x} \ge -\eps$.

Since we are interested in the case when
$\bs{\varrho} = \bs{U} \sD_\eps \bs{U}^\dag$ is random, it suffices to show
that $\bra{x} \bs{U} \sX \bs{U}^\dag \ket{x}$ concentrates in the
interval $[-\eps, \eps]$. This fact follows easily from the next lemma:
\begin{lemma}\label{lem:chi-squared-concentration}
  Let $k$ be a positive even integer. If $\bs{u} \in \CC^{k}$ is a
  uniformly random unit vector, then, for sufficiently large $k$,
  \begin{align*}
    \Pr\bracket*{\abs{\bra{\bs u} Z \ket{\bs u}} \ge \half c k^{-1/4}}
    &\le 4\exp(- \sqrt{k} c^2 / 8),
  \end{align*}
  where $Z = \diag \paren*{1, \dotsc, 1, -1, \dotsc, -1}$ is a $k \by k$
  diagonal matrix with $\tr(Z) = 0$ and $c$ may be any positive
  constant.
\end{lemma}
\begin{proof}
  Let
  $\bs{u} = (\bs{a}_1 + i \bs{b}_1, \dotsc, \bs{a}_k + i \bs{b}_k) \in
  \CC^k$ be a uniformly random unit vector with
  $\bs{a}_1, \dotsc, \bs{a}_k, \bs{b}_1, \dotsc, \bs{b}_k \in \RR$ and let
  $\bs{v} \in \RR^{2k}$ be defined by
  \begin{align*}
    \bs{v}
    &= (\bs{a}_1, \dotsc, \bs{a}_{\half[k]}, \bs{b}_1, \dotsc,
      \bs{b}_{\half[k]}, \bs{a}_{\half[k]+1}, \dotsc, \bs{a}_k,
      \bs{b}_{\half[k]+1}, \dotsc, \bs{b}_k).
  \end{align*}
  Let $D$ be the $2k \by 2k$ diagonal matrix
  $D = \diag(1, \dotsc, 1, -1, \dotsc, -1)$ with $\tr(D) = 0$. Thus,
  $\bs{v}$ is a uniformly random real unit vector such that
  $\bra{\bs v} D \ket{\bs v} = \bra{\bs u} Z \ket{\bs u}$.

  Let $\bs{x}_1, \dotsc, \bs{x}_k, \bs{y}_1, \dotsc, \bs{y}_k \in \RR$ be
  $2k$ standard Gaussian random variables. Let
  $\bs{X} = \bs{x}_1^2 + \dotsb + \bs{x}_k^2$ and
  $\bs{Y} = \bs{y}_1^2 + \dotsb + \bs{y}_k^2$. By the rotational
  symmetry of multivariate Gaussian random variables, $\bs{v}$ has the
  same distribution as
  \begin{align*}
    \frac{(\bs{x}_1, \dotsc, \bs{x}_k, \bs{y}_1, \dotsc, \bs{y}_k)}{\sqrt{\bs{X} + \bs{Y}}}.
  \end{align*}
  Hence, $\bra{\bs v} D \ket{\bs v}$ and
  $\frac{\bs{X} - \bs{Y}}{\bs{X} + \bs{Y}}$ have the same distribution.
  Since $\bs{X}$ and $\bs{Y}$ are independent $\chi^2$ random variables
  with $k$ degrees of freedom each, it holds that (see e.g.\
  \cite[Example 2.11]{Wainwright:2019})
  \begin{align*}
    \Pr\bracket*{\abs*{\frac{\bs X}{k} - 1} \ge t}
    &\le 2 \exp(-kt^2/8),
  \end{align*}
  for all $t \in (0, 1)$ and similarly for $Y$. Hence, for
  $t = ck^{-1/4}$, we have
  $\Pr\bracket*{\abs{\bs{X} - k} \ge ck^{3/4}} \le 2 \exp(- \sqrt{k} c^2
  / 8)$.

  If $\abs{\bs{X} - k} < ck^{3/4}$ and $\abs{\bs{Y} - k} < ck^{3/4}$,
  then, for $k$ sufficiently large,
  \begin{align*}
    \abs{\bra{\bs v} D \ket{\bs v}}
    &= \frac{\abs{\bs{X} - \bs{Y}}}{\bs{X} + \bs{Y}}
      \le \frac{2ck^{3/4}}{2k - 2ck^{3/4}}
      = \frac{c}{k^{1/4} - 1}
      < \half ck^{-1/4}.
  \end{align*}
  Hence,
  $\Pr[\abs{\bra{\bs v} D \ket{\bs v}} < \half ck^{-1/4}] \ge 1 - 4 \exp(-
  \sqrt{k} c^2 /8)$.
\end{proof}

If $\bs{U}$ is a random unitary distributed according to the Haar
measure on $\U(\cH)$ and $x \in \cH$ is a fixed unit vector, then
$\bs{u} = {\bs U} \ket{x}$ is a uniformly random unit vector in
$\cH$. Hence, we can apply~\Cref{lem:chi-squared-concentration} to
$\abs{\bra{\bs u} \sX \ket{\bs u}}$ to get
\begin{align}\label{eqn:concentration-ineq}
  \Pr\bracket{\abs{\bra{x} \bs{U} \sX \bs{U}^\dag \ket{x}} \ge \eps}
  &\le 4 \exp(-dc^2/8),
\end{align}
where $c$ is an arbitrary positive constant and
$\eps \ge \half c d^{-1/2}$.

We now have all the elements needed to
prove~\Cref{claim:eps-far-twothirds}:

\begin{proof}[Proof of~\Cref{claim:eps-far-twothirds}]
  Let $\bs{\varrho} = \bs{U} \sD_\eps \bs{U}^\dag$ be a uniformly random
  element of $\cD$ and let $\bs{W} = \frac{\unit}{d^2} -
  \bs{\varrho}$. Thus, assuming $\eps \ge cd^{-1/2}$,
  \begin{align*}
    &\Pr\bracket*{\forall \sigma \in \Sep, \ d_{\TV}(\bs{\varrho}, \sigma)
    \ge \eps} \\
    &\qquad=\Pr\bracket*{\forall \sigma \in \Sep, \ \norm{\bs{\varrho} - \sigma}_1
    \ge 2\eps} \\
    &\qquad\ge \Pr\bracket*{\forall \sigma \in \Sep, \ \tr(\sigma {\bs W})
      \ge - 2 \eps \norm{\bs W}_\infty}
    & (\text{by~\Cref{prop:holder-distance}}) \\
    &\qquad\ge \Pr\bracket*{\forall \sigma \in 2 \conv(\cN \cup -\cN), \ \tr(\sigma {\bs W})
      \ge - 2 \eps \norm{\bs W}_\infty}
    & (\text{by~\Cref{lem:sep-approx}}) \\
    &\qquad= \Pr\bracket*{\forall \ketbra{x}{x} \in \cN \cup -\cN, \ 2 \tr(\ketbra{x}{x} {\bs W})
      \ge - 2 \eps \norm{\bs W}_\infty}
    & (\text{by convexity}) \\
    &\qquad= \Pr\bracket*{\forall \ketbra{x}{x} \in \cN, \ \abs{\bra{x} \bs{U}
      \sX \bs{U}^\dag \ket{x}} \le \eps}
    & (\text{by~\Cref{eqn:trace-pure}}) \\
    &\qquad\ge 1 - \sum_{\ketbra{x}{x} \in \cN} \Pr\bracket*{\abs{\bra{x} \bs{U}
      \sX \bs{U}^\dag \ket{x}} > \eps}
    & (\text{by the union bound}) \\
    &\qquad\ge 1 - \abs{\cN} \cdot 4 \exp(-dc^2/8)
    & (\text{by~\Cref{eqn:concentration-ineq}}) \\
    &\qquad= 1 - 4 \exp(d (\log C - c^2/8))
    & (\text{since} \ \abs{\cN} = C^d).
  \end{align*}
  Hence, if $c = \sqrt{8 (\log C + 1)}$, then
  \begin{align*}
    \Pr\bracket*{\forall \sigma \in \Sep, \ d_{\TV}(\bs{\varrho}, \sigma)
    \ge \eps}
    &\ge 1 - 4 \exp(d (\log C - c^2/8)) \\
    &= 1 - 4 \exp(-d) \\
    &\ge \frac{2}{3},
  \end{align*}
  for $d \ge \log 12$.
\end{proof}

\bibliography{badescu}
\bibliographystyle{alpha}

\end{document}